\documentclass[11pt]{article} 
\usepackage{url}
\usepackage{graphicx,epstopdf}
\usepackage{amssymb,amsfonts,amsmath,amsthm,bm}
\usepackage{algorithm2e,framed}
\usepackage[noend]{algorithmic}
\usepackage{enumerate,wrapfig,array}

\def\math#1{$#1$}

\def\mand#1{$$#1$$}
\def\v#1{{\mathbf #1}}
\def\frac#1#2{{#1\over #2}}

\def\mld#1{\begin{equation}
#1
\end{equation}}
\def\eqar#1{\begin{eqnarray}
#1
\end{eqnarray}}
\def\eqan#1{\begin{eqnarray*}
#1
\end{eqnarray*}}

\def\cl#1{{\cal #1}}

\def\a{{\mathbf a}}
\def\b{{\mathbf b}}
\def\exp#1{{\left\langle#1\right\rangle}}

\def\norm#1{{\|#1\|}}

\def\r#1{{(\ref{#1})}}

\def\dotfil{\leaders\hbox to 1.5mm{.}\hfill}
\newcounter{rmnum}
\def\RN#1{\setcounter{rmnum}{#1}\uppercase\expandafter{\romannumeral\value{rmnum}}}
\def\rn#1{\setcounter{rmnum}{#1}\expandafter{\romannumeral\value{rmnum}}}

\DeclareMathOperator*{\argmin}{argmin}

\newcommand{\FNorm }[1]{\mbox{}\left\|#1\right\|_F  }
\newcommand{\FNormS}[1]{\mbox{}\left\|#1\right\|_F^2}

\newcommand{\setlinespacing}[1]%
           {\setlength{\baselineskip}{#1 \defbaselineskip}}

\newcommand{\rank}[1]{{\bf rank}{\left(#1\right)}}

\newtheorem{definition}{Definition}

\newtheorem{lemma}{Lemma}
\newtheorem{theorem}{Theorem}
\newtheorem{corollary}{Corollary}

\newcommand{\mat}[1]{{\ensuremath{\bm{\mathrm{#1}}}}}

\def\exp{\hbox{\rm exp}}
\def\rank{\hbox{\rm rank}}

\def\a{{\mathbf a}}

\def\b{{\mathbf b}}

\def\u{{\mathbf u}}
\def\v{{\mathbf v}}

\def\matA{\mat{A}}
\def\matB{\mat{B}}
\def\matC{\mat{C}}

\def\matE{\mat{E}}

\def\matI{\mat{I}}

\def\matU{\mat{U}}
\def\matV{\mat{V}}

\def\matX{\mat{X}}
\def\matY{\mat{Y}}

\def\matSig{\mat{\Sigma}}

\def\matPsi{\mat{\Psi}}

\def\Exp{\mathbb{E}}

\usepackage{chngcntr}

\newcounter{numrel}
\counterwithin*{numrel}{section}

\title{Column Selection via Adaptive Sampling}

\author{
Saurabh Paul
\thanks{Global Risk Sciences, Paypal Inc.
\texttt{saurabhpaul2006@gmail.com}}
\and Malik Magdon-Ismail 
\thanks{Computer Science Dept., Rensselaer Polytechnic Institute
\texttt{magdon@cs.rpi.edu}
}
\and Petros Drineas 
\thanks{Computer Science Dept., Rensselaer Polytechnic Institute
\texttt{drinep@cs.rpi.edu}} 
}

\begin{document}

\maketitle

\begin{abstract}
Selecting a good column (or row) subset of massive data matrices 
has found many applications in data analysis and machine learning. 
We propose a new adaptive sampling algorithm 
that can be used to improve \textit{any} 
relative-error column selection algorithm.  
Our algorithm delivers a tighter theoretical
bound on the approximation error which we also demonstrate empirically using
two well known relative-error column subset selection algorithms. 
Our experimental results on synthetic and real-world data show 
that our algorithm outperforms non-adaptive sampling as well as prior 
adaptive sampling approaches.
\end{abstract}

\section{Introduction}

In numerous machine learning and data analysis applications, the input data are modelled as a matrix $\matA \in \mathbb{R}^{m \times n}$, where $m$ is the number of objects (data points) and $n$ is the number of features. 
Often, it is desirable to represent your solution 
using a few features (to promote better generalization and interpretability of
the solutions), or using a few data points (to identify important coresets
of the data), for example PCA, sparse PCA, sparse regression, 
coreset based regression, etc.
\cite{malik190,malik196,malik170,malik164}.
These problems can be reduced to identifying a good subset of the 
columns (or rows) in the data matrix, the 
column subset selection problem (CSSP). For example, finding
an optimal sparse linear encoder for the data (dimension reduction)
can be explicitly reduced to CSSP~\cite{malik200}. 
Motivated by the fact that in many practical applications, the left
and right singular vectors of a matrix $\matA$ lacks any physical
interpretation, a long line of work \cite{CH92,DR10,DV06,DRVW06,DKR02,FKV98,HMT,LWMRT07,MD09, BMD14},  focused on extracting a subset of
columns of the matrix $\matA,$ which are approximately as good as $\matA_k$ at
reconstructing $\matA.$ To make our discussion more concrete, let us
formally define CSSP.\\

\noindent {\bf Column Subset Selection Problem, CSSP:} Find a matrix 
$\matC \in \mathbb{R}^{m \times c}$ containing $c$ columns of $\matA$ for which
\math{\FNorm{\matA-\matC\matC^+\matA}} is small.\footnote{\math{\matC\matC^+\matA}
is the best possible reconstruction of \math{\matA} by projection 
into the space spanned by the 
columns of~\math{\matC}.} 
In the prior work, one measures the quality of a CSSP-solution against
\math{\matA_k}, the best rank-\math{k} approximation to \math{\matA} obtained
via the singular value decomposition (SVD), where \math{k} is a 
user specified target rank parameter. For example, \cite{BMD14}
gives efficient algorithms to find \math{\matC} with 
\math{c\approx 2k/\epsilon} columns, for which
\math{\FNorm{\matA-\matC\matC^+\matA}\le (1+\epsilon)
\FNorm{\matA-\matA_k}}.\\

\noindent Our contribution is not to directly attack CSSP. We present a novel 
algorithm that can improve an existing CSSP algorithm by adaptively invoking it,
in a sense \emph{actively} learning which columns to sample next based
on the columns you have already sampled. If you use the CSSP-algorithm from 
\cite{BMD14} as a strawman benchmark, you can obtain 
\math{c} columns all at once and incur an error 
roughly \math{(1+2k/c)\FNorm{\matA-\matA_k}}.
Or, you can invoke the algorithm to obtain, for example,
\math{c/2} columns, and then 
allow the algorithm to \emph{adapt} to the
columns already chosen (for example by modifying~\math{\matA}) before
choosing
the remaining \math{c/2} columns.
We refer to the former as continued sampling and to the latter
as adaptive sampling. We prove performance guarantees which
show that adaptive sampling improves upon continued sampling, and we present
experiments on synthetic and real data that demonstrate significant
empirical performance gains.

\subsection{Notation}\label{sec:notation}
$\matA, \matB, \ldots$ denote matrices and $\a, \b, \ldots$ denote column vectors; $\matI_n$ is the  $n \times n$ identity matrix. $\left[\matA, \matB\right]$ and $\left[\matA; \matB \right]$ denote matrix concatenation operations in a column-wise and row-wise manner, respectively. Given a set $S \subseteq \{1,\ldots n\}$, $\matA_{S}$ is the matrix that contains the columns of $\matA \in \mathbb{R}^{m \times n}$ indexed by $S$. Let \math{\rank(\matA)=\rho\leq \min\left\{m,n\right\}}. The (economy) SVD of $\matA$ is
$$ \matA = \left(\matU_k\; \matU_{\rho-k}\right) \left(\begin{array}{cc}\matSig_k &\mathbf{0} \\ \mathbf{0} &\matSig_{\rho-k} \end{array}\right) \left(\begin{array}{c} \matV_k^T \\ \matV_{\rho-k}^T\end{array} \right)=\sum_{i=1}^\rho \sigma_i(\matA) \u_i\v_i^T$$
where $\matU_k \in \mathbb{R}^{m \times k}$ and $\matU_{\rho-k} \in \mathbb{R}^{m \times (\rho-k)}$ contain the left singular vectors $\u_i$, $\matV_k \in \mathbb{R}^{n \times k}$ and $\matV_{\rho-k} \in \mathbb{R}^{n \times (\rho-k)}$ contain the right singular vectors $\v_i$, and $\matSig \in \mathbb{R}^{\rho \times \rho}$ is a diagonal matrix containing the singular values $\sigma_1(\matA) \geq \ldots \geq \sigma_{\rho}(\matA) > 0.$ The Frobenius norm of $\matA$ is $\FNormS{\matA}= \sum_{i,j}\matA_{ij}^2$; $\bf{Tr}(\matA)$ is the trace of $\matA$; the 
pseudoinverse of $\matA$ is $\matA^+ = \matV \matSig^{-1} \matU^T$; and,
\math{\matA_k}, the best rank-\math{k}
approximation to \math{\matA} under any unitarily invariant norm 
is $\matA_k=\matU_k\matSig_k\matV_k^T=\sum_{i=1}^k\sigma_i\u_i\v_i^T.$

\remove{
We note that in eqn.~(\ref{eqn:pd1}), we can rewrite $\left(\matC \matC^+ \matA\right)_k$ as $\matC \matX$, where $\matX$ has rank at most $k$ (see the Appendix for the proof of the following equation):
\begin{equation}
 \matX = \argmin_{\matPsi \in \mathbb{R}^{c \times n}:\mbox{rank}\left(\matPsi\right)\leq k} \FNormS{\matA -\matC\matPsi}.
\label{eqn:eqnCX}
\end{equation}
}

\subsection{Our Contribution: Adaptive Sampling}\label{subsec:our_contri}
We design a novel CSSP-algorithm that \emph{adaptively} selects
columns from the matrix $\matA$ in \textit{rounds}.
In each round we \textit{remove} from $\matA$ the information that has already been ``captured'' by the columns that have been thus far selected. 
Algorithm~\ref{alg:alg_adaptive} selects $tc$ columns of $\matA$ in $t$ rounds, where in each round $c$ columns of $\matA$ are selected using a relative-error
CSSP-algorithm from prior work.
\\
\begin{algorithm}[htb]
\begin{center}
\parbox{0.81\textwidth}{
\begin{framed}
\textbf{Input:} $\matA \in \mathbb{R}^{m\times n}$; 
target rank $k$; \# rounds $t$; columns per round \math{c} \\
\textbf{Output:} $\matC \in \mathbb{R}^{m\times tc},$
\math{tc} columns of \math{\matA} and \math{S}, the indices of those columns.
\vskip-0.2in
\begin{algorithmic}[1]\itemsep0pt
\STATE $S=\{\}$; $\matE^0 = \matA$ 
\FOR{$\ell = 1,\cdots,t$}
\STATE Sample indices \math{S_\ell} of $c$ columns from $\matE^{\ell-1}$ 
using a CSSP-algorithm.
\STATE $S\gets S \cup S_\ell.$
\STATE Set $\matC= \matA_S$ and $\matE^{\ell} = \matA - (\matC \matC^+ \matA)_{\ell k}.$
\ENDFOR
\RETURN \math{\matC, S}
\end{algorithmic}
\end{framed}
}
\caption{Adaptive Sampling}
\label{alg:alg_adaptive}
\end{center}
\end{algorithm}

\\
At round \math{\ell} in Step 3, we compute
 column indices  $S$ (and  $\matC=\matA_S$)
using a CSSP-algorithm on the 
residual \math{\matE^{\ell-1}} of the previous round. 
To compute this residual,
remove 
from~\math{\matA} the best rank-\math{(\ell-1) k} approximation to 
\math{\matA} in the span of the  
columns selected from the first \math{\ell-1}
rounds,
\mand{\matE^{\ell-1}=\matA-(\matC \matC^+ \matA)_{(\ell-1)k}.}
A similar strategy was developed in~\cite{DV06} with sequential adaptive 
use of  
(additive error) CSSP-algorithms.
These (additive error) CSSP-algorithms select 
columns according to column norms~\cite{FKV98}. 
In~\cite{DV06}, the residual in step 5 is defined differently, as 
\math{\matE^\ell=\matA-\matC \matC^+ \matA}. To motivate our result, it 
helps to take a closer look at the reconstruction 
error \math{\matE=\matA-\matC \matC^+ \matA}
after  
\math{t} adaptive rounds of the strategy in~\cite{DV06}  
with the CSSP-algorithm in~\cite{FKV98}.
\begin{center}\begin{small}
\begin{tabular}{c|p{0.37\textwidth}|p{0.425\textwidth}}
\# rounds
&
{\bf Continued sampling: }\math{tc} columns
using CSSP-algorithm from~\cite{FKV98}.
(\math{\epsilon=k/c})
&
{\bf Adaptive sampling:} \math{t} rounds of the strategy in~\cite{DV06} with 
the CSSP-algorithm from~\cite{FKV98}.
\\\hline&&\\[-7pt]
\math{t=2}&
\multicolumn{1}{c|}{
\math{\displaystyle
\FNorm{\matE}^2
\le
\FNorm{\matA-\matA_k}^2
+
\frac{\epsilon}{2}\FNorm{\matA}^2
}
}
&
\multicolumn{1}{c}{
\math{
\FNorm{\matE}^2
\le
(1+\epsilon)\FNorm{\matA-\matA_k}^2
+
\epsilon^2\FNorm{\matA}^2
}
}
\\[8pt]
\math{t}
&
\multicolumn{1}{c|}{
\math{\displaystyle
\FNorm{\matE}^2
\le
\FNorm{\matA-\matA_k}^2
+
\frac{\epsilon}{t}\FNorm{\matA}^2
}
}
&
\math{\displaystyle
\FNorm{\matE}^2
\le
\left(1+O(\epsilon)\right)\FNorm{\matA-\matA_k}^2
+
\epsilon^t\FNorm{\matA}^2
}
\end{tabular}
\end{small}
\end{center}
Typically
\math{\FNorm{\matA}^2\gg\FNorm{\matA-\matA_k}^2} and  
\math{\epsilon} is small (i.e., \math{c\gg k}), so adaptive sampling 
\`{a} la~\cite{DV06} wins over continued sampling
for additive error CSSP-algorithms. This is especially apparent after 
\math{t} rounds, where continued sampling only attenuates the big term
\math{\norm{\matA}_F^2} by \math{\epsilon/t}, but adaptive sampling 
exponentially attenuates this term by \math{\epsilon^t}.

Recently, powerful CSSP-algorithms
have been developed which give relative-error 
guarantees~\cite{BMD14}. 
We can use the adaptive strategy from~\cite{DV06}
together with 
these newer relative error CSSP-algorithms. If one carries out the analysis
from~\cite{DV06} by replacing the additive error
CSSP-algorithm from ~\cite{FKV98} with 
the relative error CSSP-algorithm in~\cite{BMD14}, the comparison of 
continued and adaptive sampling using the strategy from ~\cite{DV06} becomes 
(\math{t=2} rounds suffices to see the problem):
\begin{center}\begin{small}
\begin{tabular}{c|p{0.37\textwidth}|p{0.425\textwidth}}
\# rounds
&
{\bf Continued sampling:} \math{tc} columns
using CSSP-algorithm from~\cite{BMD14}.
(\math{\epsilon=2k/c})
&
{\bf Adaptive sampling:} \math{t} rounds of the strategy in~\cite{DV06} with 
the CSSP-algorithm from~\cite{BMD14}.
\\\hline&&\\[-7pt]
\math{t=2}&
\multicolumn{1}{c|}{
\math{\displaystyle
\FNorm{\matE}^2
\le
\left(1+\frac\epsilon2\right)\FNorm{\matA-\matA_k}^2
}
}
&
\multicolumn{1}{c}{
\math{\displaystyle
\FNorm{\matE}^2
\le
\left(1+\frac\epsilon2+\frac{\epsilon^2}2\right)\FNorm{\matA-\matA_k}^2
}
}
\end{tabular}
\end{small}
\end{center}
Adaptive sampling from
\cite{DV06} gives a worse theoretical guarantee than continued sampling for  
relative error CSSP-algorithms. In a nutshell, no matter how many
rounds of adaptive sampling you do, the theoretical bound will not be
better than \math{(1+k/c)\norm{\matA-\matA_k}_F^2} if you are using a
relative error CSSP-algorithm. This 
raises an obvious question: is it possible to combine 
relative-error CSSP-algorithms with adaptive sampling 
to get (provably and empirically) improved CSSP-algorithms? 
The approach of~\cite{DV06} does not achieve this objective.
We provide a positive answer to this question.

Our approach is a
subtle modification to the approach 
in~\cite{DV06}: in Step 5 of Algorithm~\ref{alg:alg_adaptive}.
When we compute the residual matrix in round \math{\ell},
we subtract $(\matC \matC^+ \matA)_{\ell k}$ from \math{\matA}, the best 
rank-\math{\ell k} approximation to the projection of \math{\matA}
onto the current columns selected, 
as opposed to subtracting the full projection 
$\matC \matC^+ \matA$. 
This subtle change, is critical in our new analysis which gives 
a tighter bound on the final error, 
allowing us to boost relative-error 
CSSP-algorithms. For \math{t=2} rounds of adaptive sampling, we get a
reconstruction error of 
\mand{
\norm{\matE}_F^2 
\leq (1+\epsilon) \FNormS{\matA - \matA_{2k}}+ \epsilon (1+\epsilon)\FNormS{\matA - \matA_{k}},
}
where \math{\epsilon=2k/c}.
The critical improvement in the bound is that the dominant \math{O(1)}-term
depends on
\math{\norm{\matA - \matA_{2k}}_F^2}, and the dependence on 
\math{\norm{\matA - \matA_{k}}_F^2} is now \math{O(\epsilon)}.
To highlight this improved theoretical bound in an extreme case, consider a matrix $\matA$ that has rank exactly $2k$, then \math{\norm{\matA - \matA_{2k}}_F=0}.
Continued sampling gives an error-bound
\math{(1+\frac{\epsilon}{2}) \norm{\matA-\matA_k}_F^2}, where as
our adaptive sampling gives an error-bound
\math{(\epsilon+\epsilon^2) \norm{\matA-\matA_k}_F^2}, which is clearly
better in this extreme case. 
In practice, data matrices have rapidly decaying singular values, so this 
extreme case is not far from reality (See Figure \ref{fig:spectrum}).

\begin{figure}[!htb]
\begin{center}
\includegraphics[height = 38mm,width=0.45\columnwidth,clip]{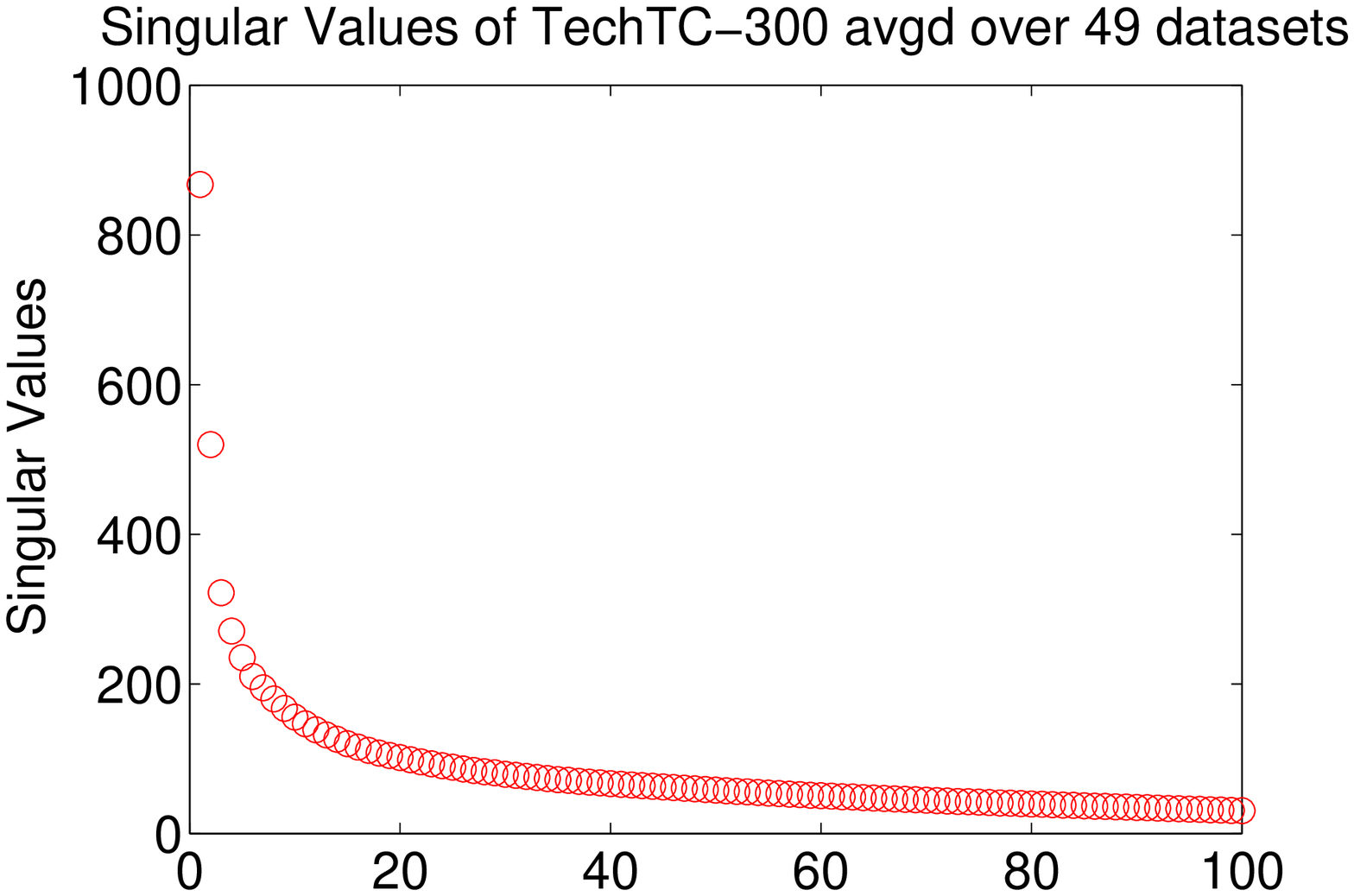}
\includegraphics[height = 38mm,width=0.45\columnwidth,clip]{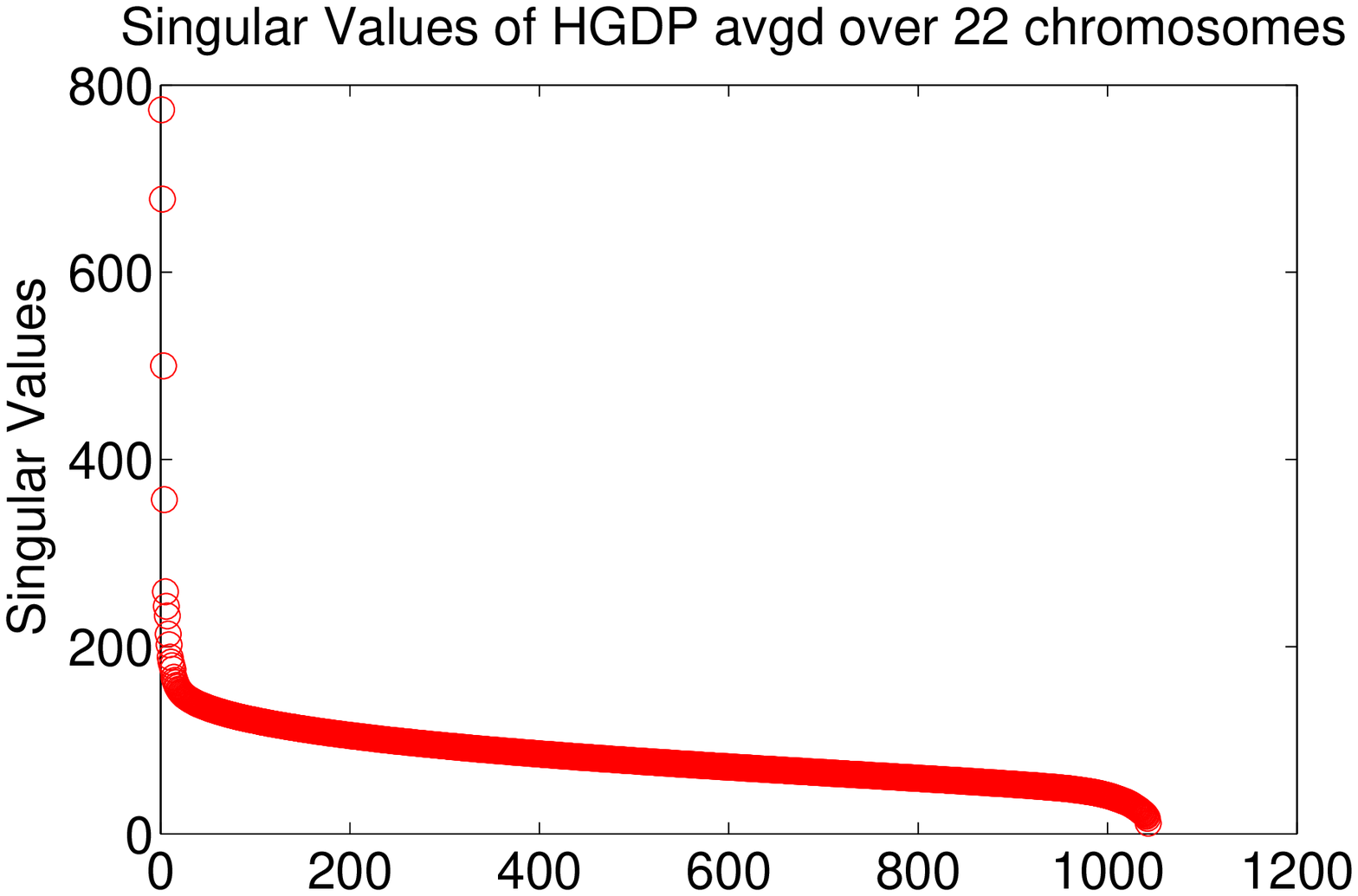}
\end{center}
\caption{\small Figure showing the singular value decay for two real world datasets.}\vskip -0.2cm
\label{fig:spectrum}
\end{figure}

To state our main theoretical result, we need to more formally
define a relative error CSSP-algorithm.
\begin{definition}[Relative Error CSSP-algorithm \math{\cl{A}(\matX,k,c)}]
\label{def:relCSSP}
A relative error CSSP-algorithm \math{\cl{A}} 
takes as input a matrix 
\math{\matX}, a rank parameter \math{k<\rank(\matX)} and a number of columns
\math{c},
and outputs column indices \math{S} with \math{|S|=c},
so that
the columns \math{\matC=\matX_S}
satisfy:
\mand{
\Exp_{\matC}\left[\norm{\matX-(\matC\matC^+\matX)_k}_F^2\right]
\le
(1+\epsilon(c,k)) \norm{\matX-\matX_k}_F^2,
}
where \math{\epsilon(c,k)} depends on \math{\cl{A}} and
the expectation is over random choices made in the algorithm.\footnote{For an additive-error CSSP algorithm,
$\Exp_{\matC}\left[\norm{\matX-(\matC\matC^+\matX)_k}_F^2\right]
\le
\norm{\matX-\matX_k}_F^2+ \epsilon(c,k)\norm{\matX}_F^2.$}
\end{definition}
Our main theorem bounds the reconstruction error 
when our adaptive sampling approach is used to boost \math{\cl{A}}.
The boost in performance depends on the decay of the spectrum of $\matA$.
\begin{theorem}
Let $\matA \in \mathbb{R}^{m\times n}$ be a matrix of rank $\rho$ and let $k <\rho$ be a target rank. If, in Step 3 of Algorithm~\ref{alg:alg_adaptive}, we
use the relative error CSSP-algorithm \math{\cl{A}} with 
\math{\epsilon(c,k)=\epsilon<1}, then
$$\Exp_{\matC}\left[\norm{\matA - (\matC\matC^+\matA)_{tk}}_F^2\right] \leq (1+\epsilon) \FNormS{\matA - \matA_{tk}}+ \epsilon \sum_{i=1}^{t-1} (1+\epsilon)^{t-i}\FNormS{\matA - \matA_{ik}}.$$
\label{thm:thm1}
\end{theorem}
{\bf Comments.}
\begin{enumerate}[{\bf 1.}]\itemsep0pt
\item 
The dominant \math{O(1)}
term in our bound is $\FNorm{\matA-\matA_{tk}}$, 
not $\FNorm{\matA-\matA_k}$. This is a major improvement 
since the former is typically \textit{much smaller} than the latter
in real data. 
Further, we need a bound on the reconstruction error 
\math{\norm{\matA -\matC\matC^+\matA}_F}.
Our theorem give a stronger result than
needed because
\math{\norm{\matA -\matC\matC^+\matA}_F
\le\norm{\matA - (\matC\matC^+\matA)_{tk}}_F}.
\item 
We presented our result for the case of a relative error CSSP-algorithm
with a guarantee on the expected reconstruction error. Clearly, if the
CSSP-algorithm is deterministic, then Theorem~\ref{thm:thm1} will also hold
deterministically. The result in Theorem~\ref{thm:thm1} can also be boosted
to hold with high probability, by repeating the 
process \math{\log\frac1\delta} times and picking the columns which 
performed best. Then, with probability at least \math{1-\delta}, 
$$\norm{\matA - (\matC\matC^+\matA)_{tk}}_F^2 \leq (1+2\epsilon) \FNormS{\matA - \matA_{tk}}+ 2\epsilon \sum_{i=1}^{t-1} (1+\epsilon)^{t-i}\FNormS{\matA - \matA_{ik}}.$$
If the CSSP-algorithm itself only gives a high-probability (at least
\math{1-\delta}) guarantee, then the bound in 
Theorem~\ref{thm:thm1} also holds with high probability, at least
\math{1-t\delta}, which is obtained by applying a union bound
to the probability of failure in each round.
\item
Our results hold for any relative error CSSP-algorithm combined 
with our adaptive 
sampling strategy. 
The relative error CSSP-algorithm in \cite{BMD14} has 
\math{\epsilon(c,k)\approx 2k/c}. The relative error CSSP-algorithm in
\cite{DMM06ar} has \math{\epsilon(c,k)=O(k\log k/c)}. 
Other algorithms can be found in~\cite{DV06,DRVW06,Sinop}. We presented the 
simplest form of the result, which can be generalized to sample a different
number of columns in each round, or even use a different CSSP-algorithm in 
each round. We have not optimized the sampling schedule, i.e. how many columns
to sample in each round.
At the moment, this is largely dictated by the CSSP algorithm itself, which
requires a minimum number of samples in each round to give a theoretical
guarantee. From the empirical perspective (for example using
leverage score sampling to select columns), 
strongest performance may be obtained by adapting after every column 
is selected.
\item 
In the context of the additive error 
CSSP-algorithm from~\cite{FKV98}, our adaptive sampling strategy
 gives a theoretical performance 
guarantee which is at least as good as the adaptive sampling strategy
from~\cite{DV06}.
\end{enumerate}
Lastly, we also provide the \textit{first} empirical evaluation of adaptive sampling algorithms. We implemented our algorithm using two relative-error column selection algorithms (the near-optimal column selection algorithm of~\cite{BMD11,BMD14} and the leverage-score sampling algorithm of~\cite{DrineasrelativeCUR}) and compared it against the adaptive sampling algorithm of~\cite{DV06} on synthetic and real-world data. The experimental results show that our algorithm outperforms prior approaches.

\subsection{Related Work} \label{subsec:rw}

Column selection algorithms have been extensively studied in prior literature. Such algorithms include rank-revealing QR factorizations \cite{CH92,ChanQR}
for which only weak performance guarantees can be derived. The QR approach was improved in \cite{maung2013pass} where the authors proposed a memory efficient implementation. The randomized additive error CSSP-algorithm~\cite{FKV98} was a breakthrough, 
which led to a series of improvements producing relative CSSP-algorithms using
a variety of randomized and deterministic techniques.
These include leverage score sampling~\cite{DrineasrelativeCUR,DMM06ar}, volume sampling~\cite{DV06,DRVW06,Sinop}, the two-stage hybrid sampling approach of~\cite{BMD09}, the near-optimal column selection algorithms of~\cite{BMD11,BMD14}, as well as deterministic variants presented in~\cite{Boutsi14}.
We refer the reader to Section~1.5 of~\cite{BMD14} for a detailed overview of prior work. Our focus is not on CSSP-algorithms \emph{per se}, but rather
on adaptively invoking existing CSSP-algorithms. The only 
prior adaptive sampling with a provable guarantee 
was introduced in~\cite{DV06} and 
further analyzed in~\cite{Deshpande06soda,DRVW06,DM07}; this strategy 
is specifically boosts the additive error CSSP-algorithm, but does not
work with relative error CSSP-algorithms which are currently in use.
Our modification of the approach in~\cite{DV06}
is delicate, but crucial to the new
analysis we perform in the context of relative error CSSP-algorithms.
 
Our work is motivated by relative error CSSP-algorithms satisfying
definition~\ref{def:relCSSP}. Such algorithms exist which give 
expected guarantees~\cite{BMD14} as well as 
high probability guarantees~\cite{DrineasrelativeCUR}. 
Specifically, 
given $\matX \in \mathbb{R}^{m\times n}$ and a target rank $k$, 
the leverage-score sampling approach of~\cite{DrineasrelativeCUR} selects 
$c=O\left(\left(k/\epsilon^2\right)\log\left(k/\epsilon^2\right)\right)$
columns of $\matA$ to form a matrix $\matC \in \mathbb{R}^{m\times c}$ 
to give a \math{(1+\epsilon)}-relative error  
with probability at least $1-\delta$. Similarly,~\cite{BMD11,BMD14} 
proposed near-optimal relative error CSSP-algorithms selecting 
\math{c\approx 2c/\epsilon} columns and giving a 
\math{(1+\epsilon)}-relative error guarantee in expectation, 
which can be boosted to a
high probability guarantee via independent repetition.

\section{Proof of Theorem~\ref{thm:thm1}}

We now prove the main result which analyzes 
the performance of our adaptive sampling in 
Algorithm~\ref{alg:alg_adaptive} for a relative error CSSP-algorithm.
We will need the following linear algebraic 
Lemma.
\begin{lemma}\label{lem:lem1_adp}
Let $\matX,\matY \in \mathbb{R}^{m\times n}$ and suppose that
\math{\rank(\matY)=r}. Then, 
\mand{
\sigma_i(\matX-\matY)\ge \sigma_{r+i}(\matX).
}
\end{lemma}
\begin{proof}
Observe that
\math{\sigma_i(\matX-\matY)=\norm{(\matX-\matY)-(\matX-\matY)_{i-1}}_2}.
The claim is now 
immediate from the Eckart-Young theorem because 
\math{\matY+(\matX-\matY)_{i-1}} has rank at most 
\math{r+i-1}, therefore
\mand{\sigma_i(\matX-\matY)=\norm{\matX-(\matY+(\matX-\matY)_{i-1})}_2
\ge \norm{\matX-\matX_{r+i-1}}_2=\sigma_{r+i}(\matX).}
\end{proof}
\remove{
\begin{corollary}\label{cor:lem1_adp}
Let $\matX,\matY \in \mathbb{R}^{m\times n}$ and suppose that
\math{\rank(\matY)=r}. Then,
\mand{
\FNormS{\matX-\matX_{\tau}} - \sum_{i=1}^r \sigma_i^2 (\matX-\matY) \leq 
\FNormS{\matX-\matX_{\tau+r}}.
}
\end{corollary}
\begin{proof} Let \math{\rho=\rank(\matX)}. Then,
\eqan{
\FNormS{\matX-\matX_{\tau}}-\sum_{i=1}^r \sigma_i^2 (\matX-\matY) 
&=&
\sum_{i=\tau+1}^{\rho}\sigma_i^2(\matX)-\sum_{i=1}^r \sigma_i^2 (\matX-\matY) 
\\
&\le&
\sum_{i=\tau+1}^{\rho}\sigma_i^2(\matX)-\sum_{i=1}^r \sigma_{\tau+i}^2 (\matX)
\\
&=&
\sum_{i=\tau+r+1}^{\rho}\sigma_i^2(\matX)
=
\FNormS{\matX-\matX_{\tau+r}}.
}
(The inequality follows from Lemma~\ref{lem:lem1_adp}.)
\end{proof}
}
We are now ready to prove Theorem~\ref{thm:thm1}
by induction on \math{t}, the number of rounds of adaptive sampling. When
\math{t=1}, the claim is that
$$\Exp\left[\norm{\matA - (\matC\matC^+\matA)_{k}}_F^2\right] 
\leq (1+\epsilon) \FNormS{\matA - \matA_{k}},$$
which is immediate from  the 
definition of the relative error CSSP-algorithm.
Now for the induction.
Suppose that after \math{t} rounds, columns \math{\matC^t} are selected,
and we have the induction hypothesis that 
\mld{
\Exp_{\matC^t}\left[\norm{\matA - (\matC^t\matC^{t+}\matA)_{tk}}_F^2\right] 
\leq (1+\epsilon) \FNormS{\matA - \matA_{tk}}+ \epsilon \sum_{i=1}^{t-1} (1+\epsilon)^{t-i}\FNormS{\matA - \matA_{ik}}.
\label{eq:indhyp}
}
In the \math{(t+1)}th round, we use the residual 
\math{\matE^t=\matA - (\matC^t\matC^{t+}\matA)_{tk}} to select new columns 
\math{\matC'}. Our relative error CSSP-algorithm \math{\cl{A}} 
gives the following guarantee:
\eqar{
\Exp_{\matC'}\left[\left.
\norm{\matE^t - (\matC'{\matC'}^+\matE^t)_{k}}_F^2\right|\matE^t\right] 
&\leq& 
(1+\epsilon) \FNormS{\matE^t - \matE^t_{k}} \nonumber\\
&=&
(1+\epsilon)\left(\FNormS{\matE^t}-\sum_{i=1}^k\sigma_i^2(\matE^t)\right)
\nonumber\\
&\leq& 
(1+\epsilon)\left(\FNormS{\matE^t}-\sum_{i=1}^k\sigma_{tk+i}^2(\matA)\right).
\label{eq:crucial}
}
(The last step follows because 
\math{\sigma_i^2(\matE^t)=\sigma_i^2(\matA - (\matC^t\matC^{t+}\matA)_{tk})}
and we can apply 
Lemma~\ref{lem:lem1_adp} with \math{\matX=\matA}, 
\math{\matY=(\matC^t\matC^{t+}\matA)_{tk}} and \math{r=\rank(\matY)=tk},
to obtain \math{\sigma_i^2(\matE^t)\ge \sigma_{tk+i}^2(\matA)}.)
We now take the expectation of both sides with respect to the columns 
\math{\matC^t},
\eqar{
&&
\Exp_{\matC^t}\left[\Exp_{\matC'}\left[\left.
\norm{\matE^t - (\matC'{\matC'}^+\matE^t)_{k}}_F^2\right|\matE^t\right] 
\right]
\nonumber\\
&\le&
(1+\epsilon)\left(\Exp_{\matC^t}\left[\FNormS{\matE^t}\right]-\sum_{i=1}^k\sigma_{tk+i}^2(\matA)\right).
\nonumber\\
&\buildrel (a)\over \le&
(1+\epsilon)^2\norm{\matA - \matA_{tk}}_F^2
+
\epsilon \sum_{i=1}^{t-1} (1+\epsilon)^{t+1-i}\FNormS{\matA - \matA_{ik}}
-(1+\epsilon)\sum_{i=1}^k\sigma_{tk+i}^2(\matA)
\nonumber\\
&=&
(1+\epsilon)\left(\norm{\matA - \matA_{tk}}_F^2-\sum_{i=1}^k\sigma_{tk+i}^2(\matA)
\right)
+
\epsilon(1+\epsilon)\norm{\matA - \matA_{tk}}_F^2
\nonumber\\
&&+
\epsilon \sum_{i=1}^{t-1} (1+\epsilon)^{t+1-i}\FNormS{\matA - \matA_{ik}}
\nonumber\\
&=&
(1+\epsilon)\norm{\matA - \matA_{(t+1)k}}_F^2+
\epsilon \sum_{i=1}^{t} (1+\epsilon)^{t+1-i}\FNormS{\matA - \matA_{ik}}
\label{eq:RHS}
}
(a) follows, because of the induction hypothesis (eqn.~\ref{eq:indhyp}).
The columns chosen after round \math{t+1} are 
\math{\matC^{t+1}=[\matC^t,\matC']}. By the law of iterated expectation,
\mand{
\Exp_{\matC^t}\left[\Exp_{\matC'}\left[\left.
\norm{\matE^t - (\matC'{\matC'}^+\matE^t)_{k}}_F^2\right|\matE^t\right] 
\right]
=
\Exp_{\matC^{t+1}}\left[
\norm{\matE^t - (\matC'{\matC'}^+\matE^t)_{k}}_F^2\right].
}
Observe that 
\math{\matE^t - (\matC'{\matC'}^+\matE^t)_{k}
=
\matA - (\matC^t\matC^{t+}\matA)_{tk}- (\matC'{\matC'}^+\matE^t)_{k}
=
\matA-\matY}, where \math{\matY} is in the column space 
of \math{\matC^{t+1}=[\matC^t,\matC']};
further,
\math{\rank(\matY)\le(t+1)k}. Since
\math{(\matC^{t+1}{\matC^{t+1}}^+\matA)_{(t+1)k}} is the best 
rank-\math{(t+1)k} approximation to \math{\matA} in the column space 
of \math{\matC^{t+1}}, for any realization of
\math{\matC^{t+1}},
\mld{
\norm{\matA-(\matC^{t+1}{\matC^{t+1}}^+\matA)_{(t+1)k}}_F^2\le
\norm{\matE^t - (\matC'{\matC'}^+\matE^t)_{k}}_F^2.
\label{eq:LHS}
}
Combining \r{eq:LHS} with \r{eq:RHS}, we have that
\begin{eqnarray}
&&\Exp_{\matC^{t+1}}\left[\norm{\matA-(\matC^{t+1}{\matC^{t+1}}^+\matA)_{(t+1)k}}_F^2
\right] \nonumber \\
&\le&
(1+\epsilon)\norm{\matA - \matA_{(t+1)k}}_F^2+
\epsilon \sum_{i=1}^{t} (1+\epsilon)^{t+1-i}\FNormS{\matA - \matA_{ik}}.\nonumber
\end{eqnarray}
This is the desired bound after \math{t+1} rounds, concluding the induction.
\qed

It is instructive to understand where our new adaptive
sampling strategy
is needed for the proof to go through.
The crucial step is \r{eq:crucial} where we use 
Lemma~\ref{lem:lem1_adp} -- it is essential that
the residual was a low-rank perturbation of \math{\matA}.

\section{Experiments}
We compared three adaptive column sampling
methods, using two real and two synthetic data sets.\footnote{{\bf ADP-Nopt:}
has two stages. The first stage is a 
deterministic dual set spectral-Frobenius column selection in which ties could
occur. We break ties in favor of the column not already selected with the 
maximum norm.}
\\
\underline{Adaptive Sampling Methods}\\[2pt]
{\bf ADP-AE:} The prior adaptive method 
which uses the additive error CSSP algorithm~\cite{DV06}.
\\
{\bf ADP-LVG:} Our new adaptive method using the relative error 
CSSP algorithm~\cite{DrineasrelativeCUR}.
\\
{\bf ADP-Nopt:} Our adaptive method using the near optimal
relative error 
CSSP algorithm~\cite{BMD14}.
\\
%
\underline{Data Sets}\\[2pt]
{\bf HGDP 22 chromosomes:}
SNPs human chromosome data from the HGDP database \cite{Pasch10}. 
We use all 22 chromosome matrices (1043 rows; 7,334-37,493 columns) 
and report the average.
Each matrix contains $+1,0,-1$ entries, and we randomly filled in
missing entries.
\\
{\bf TechTC-300:}
49 document-term matrices \cite{David04} (150-300 rows (documents);  
10,000-40,000 columns (words)). We kept 5-letter or larger words and
report averages over 49 data-sets.
\\
{\bf Synthetic 1:} Random \math{1000\times 10000} matrices with 
$\sigma_i = i^{-0.3}$ (power law).
\\
{\bf Synthetic 2:} Random \math{1000\times 10000} matrices with 
$\sigma_i=\exp^{(1-i)/10}$ (exponential).
\\
\begin{figure}[t]
\begin{center}
\includegraphics[width=0.45\columnwidth,clip]{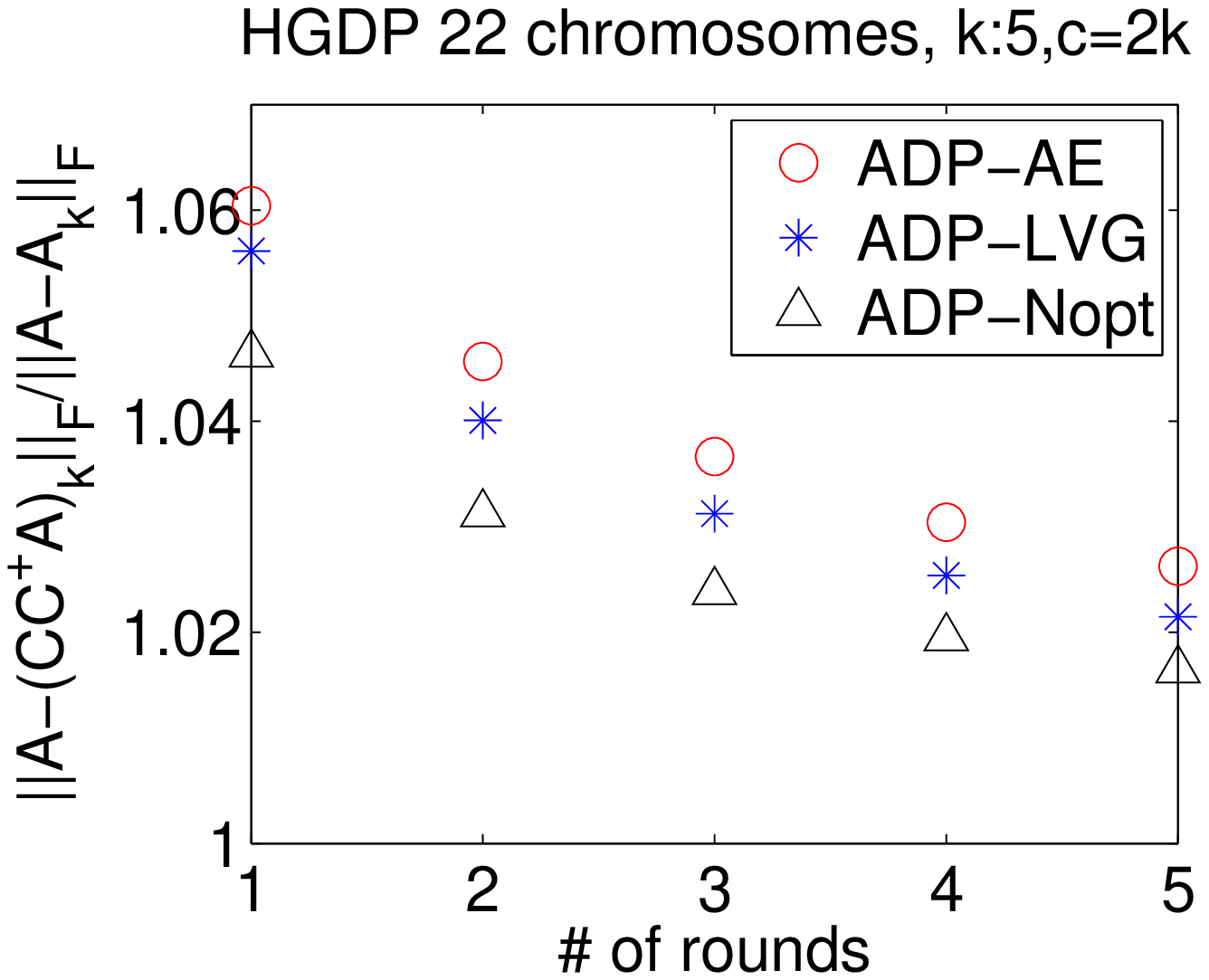}
\hspace*{0.25in}
\includegraphics[width=0.45\columnwidth,clip]{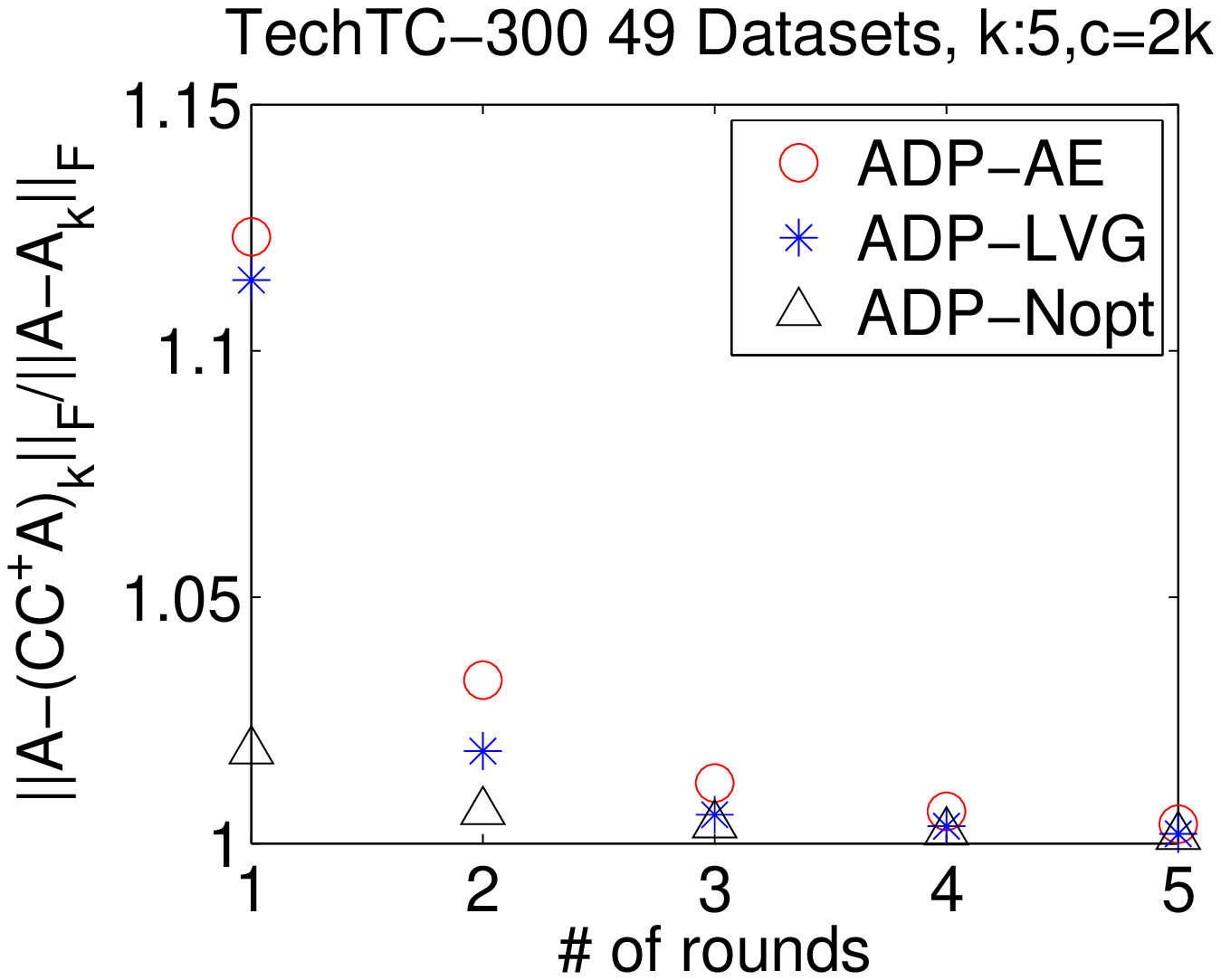}
\\[5pt]
\includegraphics[width=0.45\columnwidth,clip]{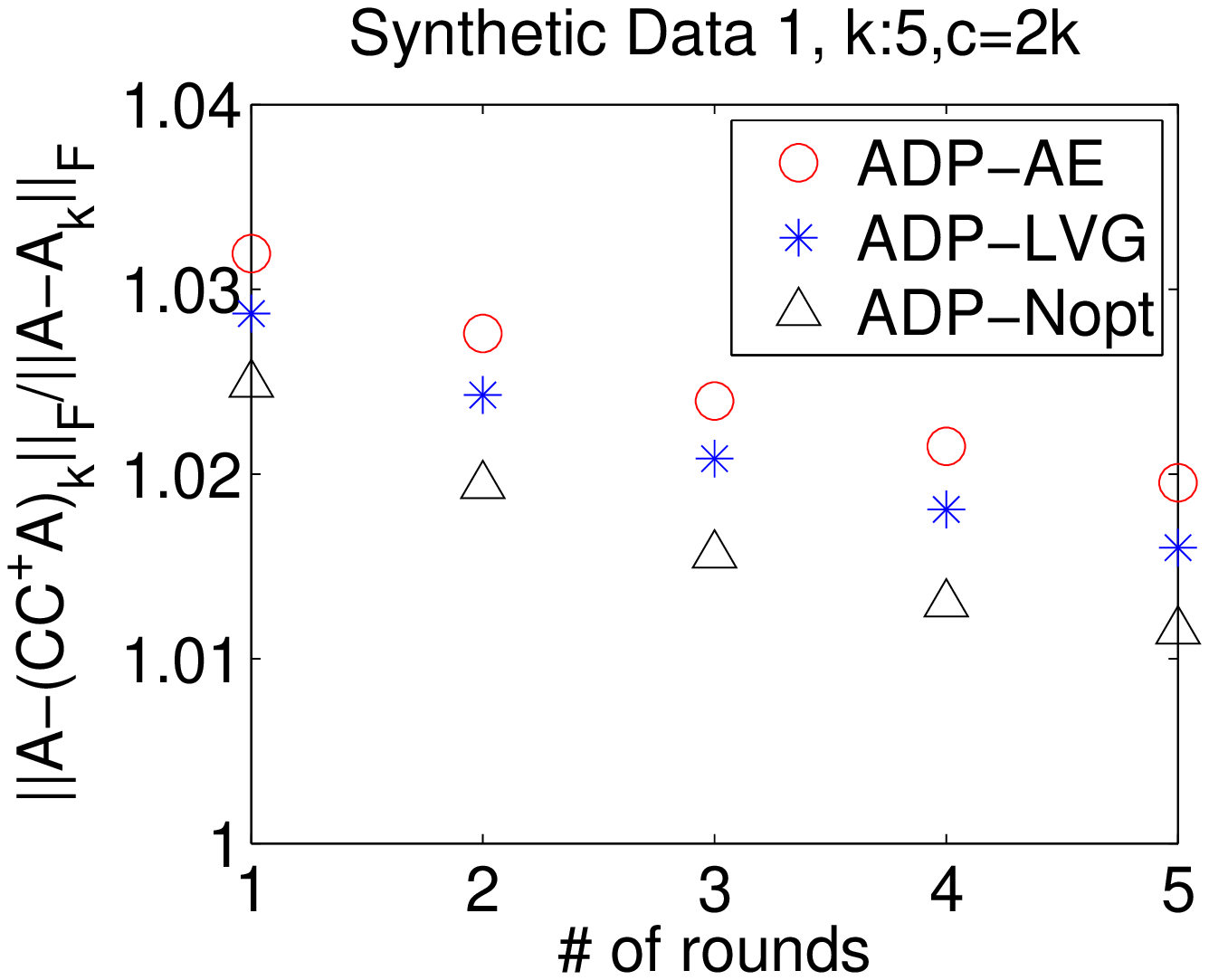}
\hspace*{0.25in}
\includegraphics[width=0.45\columnwidth,clip]{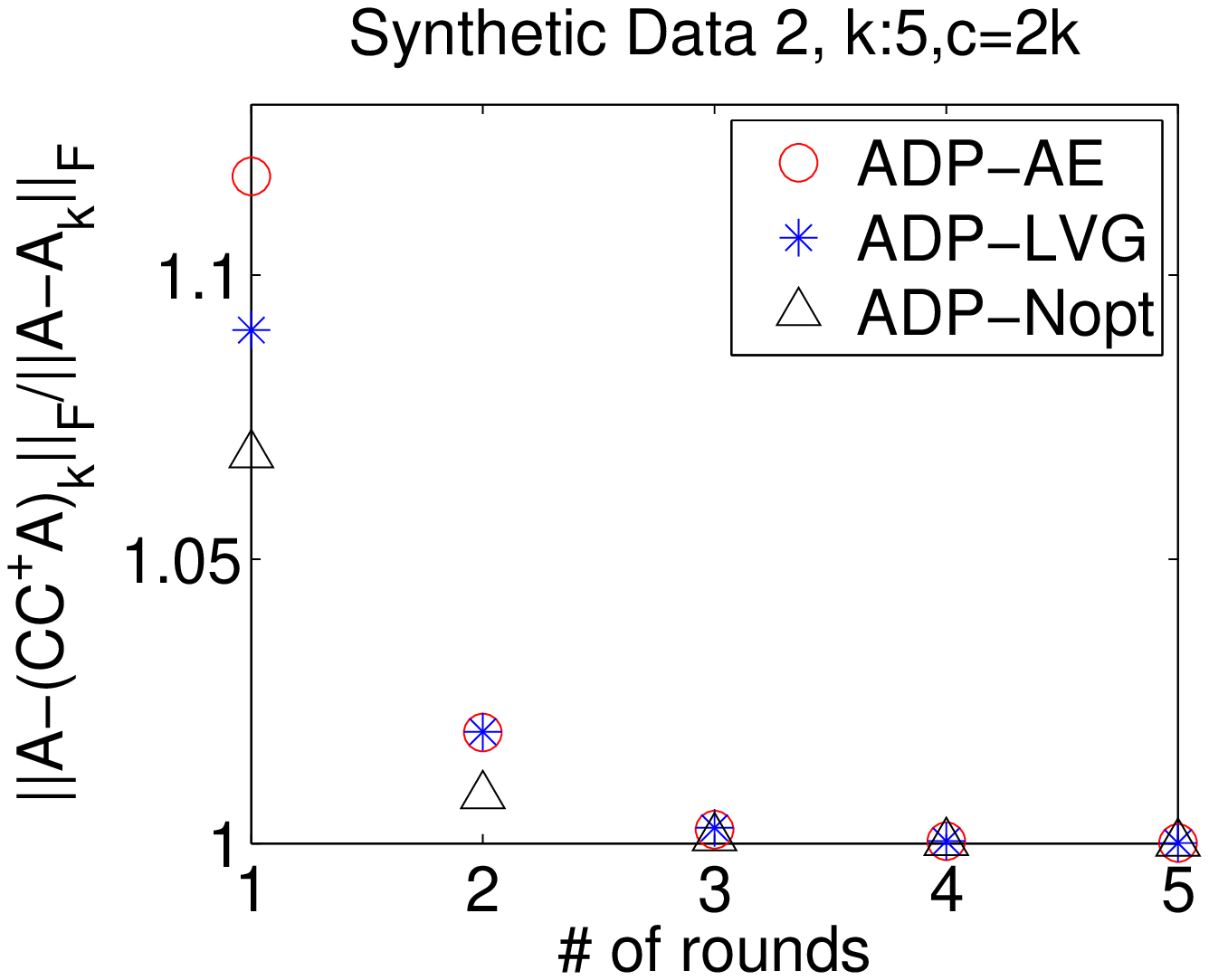} 
\end{center} \vskip -0.2cm
\caption{\small Plots of relative error ratio 
$\FNorm{\matA -(\matC\matC^+\matA)_k}/\FNorm{\matA -\matA_k}$
for various adaptive sampling algorithms for \math{k=5} and
 $c=2k.$ In all cases, performance improves with 
more rounds of sampling, and rapidly converges to a relative 
reconstruction error of 1. This is most so in data matrices with
singular values that decay quickly (such as TectTC and Synthetic 2). 
The HGDP singular values decay slowly because missing entries are selected randomly, and Synthetic 1 has slowly decaying power-law singular values by construction.}
\label{fig:plots}
\end{figure}

For randomized algorithms, we repeat the experiments five times and 
take the average. 
We use the synthetic data sets to provide a controlled environment 
in which we can see performance for different types of singular value
spectra on very large matrices.
In prior work it is common to report on the quality of the columns selected
\math{\matC} by comparing the best rank-\math{k} approximation within 
the column-span of \math{\matC} to \math{\matA_k}. Hence, we report 
the relative error
$\FNorm{\matA -(\matC\matC^+\matA)_k}/\FNorm{\matA -\matA_k}$ when 
comparing the algorithms. We set the target rank $k=5$ and 
the number of columns in each round to \math{c=2k}.
We have tried several choices for \math{k} and \math{c} and the results 
are qualitatively identical so we only report on one choice.
Our first set of results in Figure~\ref{fig:plots} is to compare the prior
adaptive algorithm {\bf ADP-AE} with the new adaptive ones 
{\bf ADP-LVG} and 
{\bf ADP-Nopt} which boose relative error CSSP-algorithms. Our two new 
algorithms are both performing better the prior existing adaptive sampling
algorithm. Further, {\bf ADP-Nopt} is performing better than {\bf ADP-LVG}, and
this is also not surprising, because  {\bf ADP-Nopt} produces near-optimal
columns -- if you boost a better CSSP-algorithm, you get better results.
Further, by comparing the performance on Synthetic 1 with Synthetic 2, we see
that our algorithm (as well as prior algorithms) gain significantly
in performance for rapidly decaying singular values; our new theoretical
analysis reflects this behavior, whereas prior results do not.

\begin{figure}[!htb]
\begin{center}
\includegraphics[height=45mm,width=0.45\columnwidth,clip]{hgdp_c2k_k5.eps}
\includegraphics[height=45mm,width=0.45\columnwidth,clip]{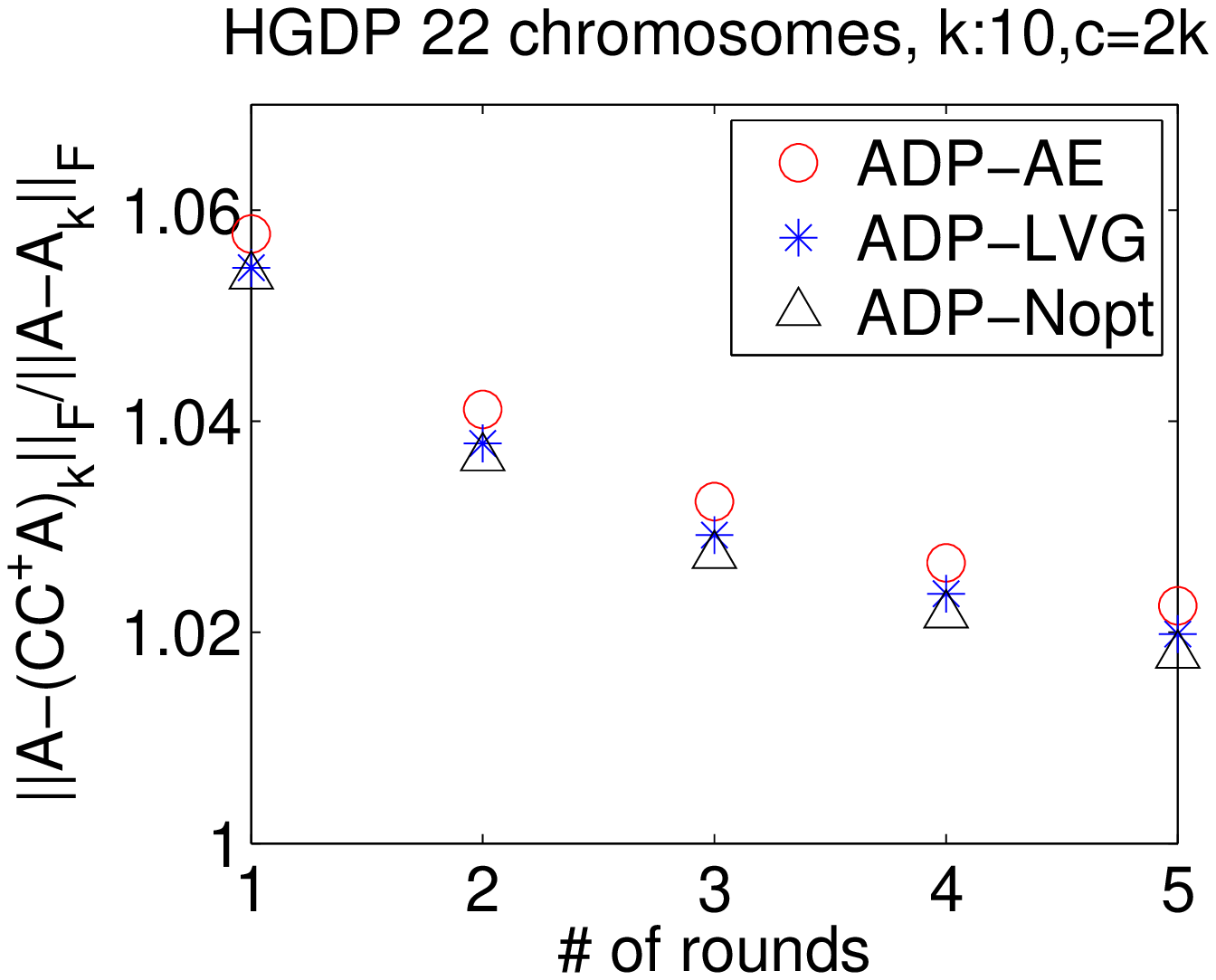}
\end{center}
\caption{Plots showing the error ratio as a function of $k.$}
\label{fig:plots2}
\end{figure}
The theory bound depends on 
\math{\epsilon=c/k}. 
Figure~\ref{fig:plots2}
shows a result for \math{k=10}; \math{c=2k} (\math{k} increases but
\math{\epsilon} is constant). We see 
that the quantitative performance is approximately 
the same, as the theory predicts (since \math{c/k} has not
changed). The percentage error stays the same even though
we are sampling \emph{more} columns because the benchmark 
\math{\norm{\matA-\matA_k}_F} also get smaller when \math{k} increases.
Since {\bf ADP-Nopt} is the superior algorithm, we  
continue with results only for this algorithm.

\begin{figure}[!htb]
\begin{center}
\includegraphics[height=45mm, width=0.45\columnwidth,clip]{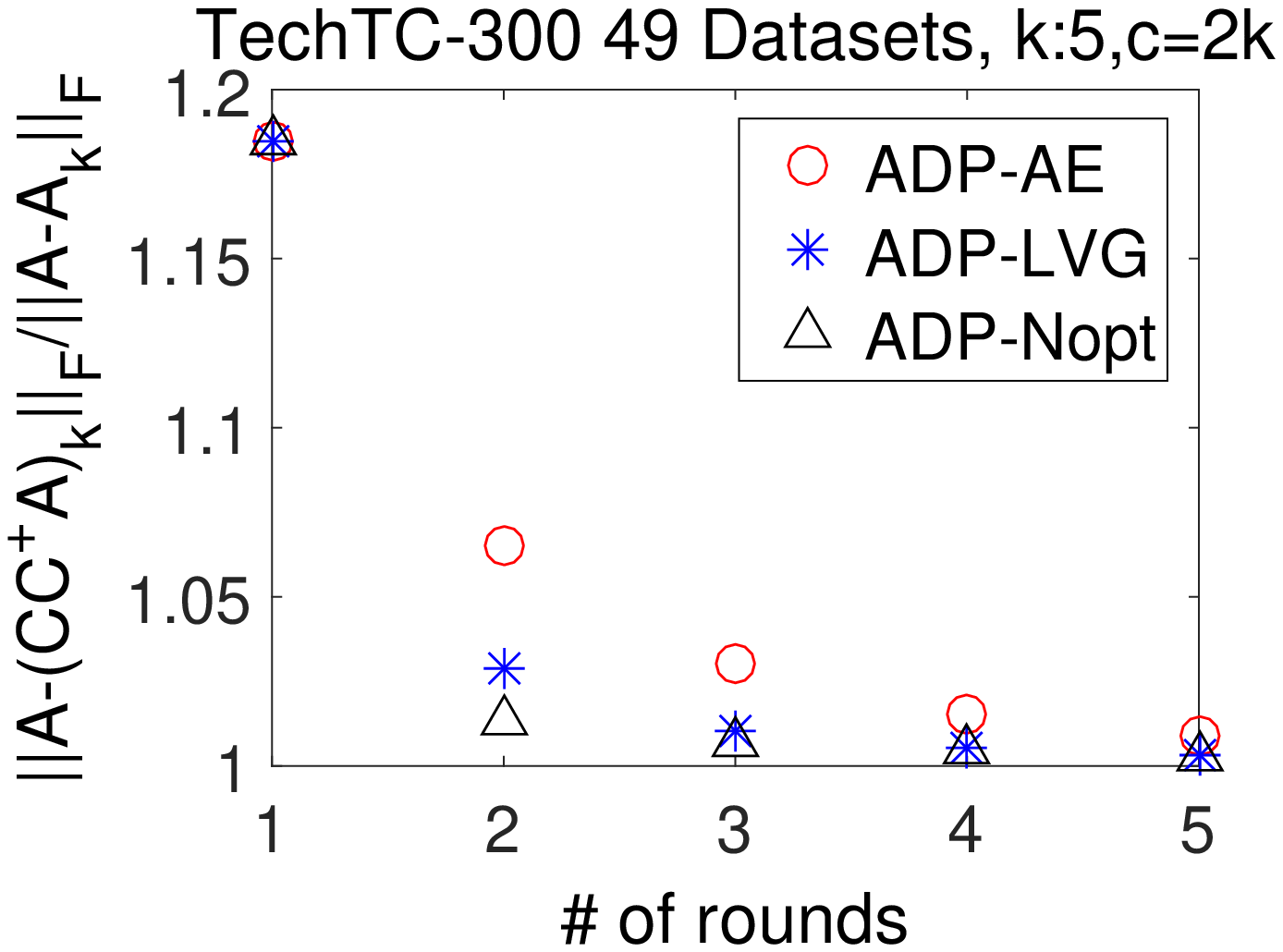}
\includegraphics[height=45mm, width=0.45\columnwidth,clip]{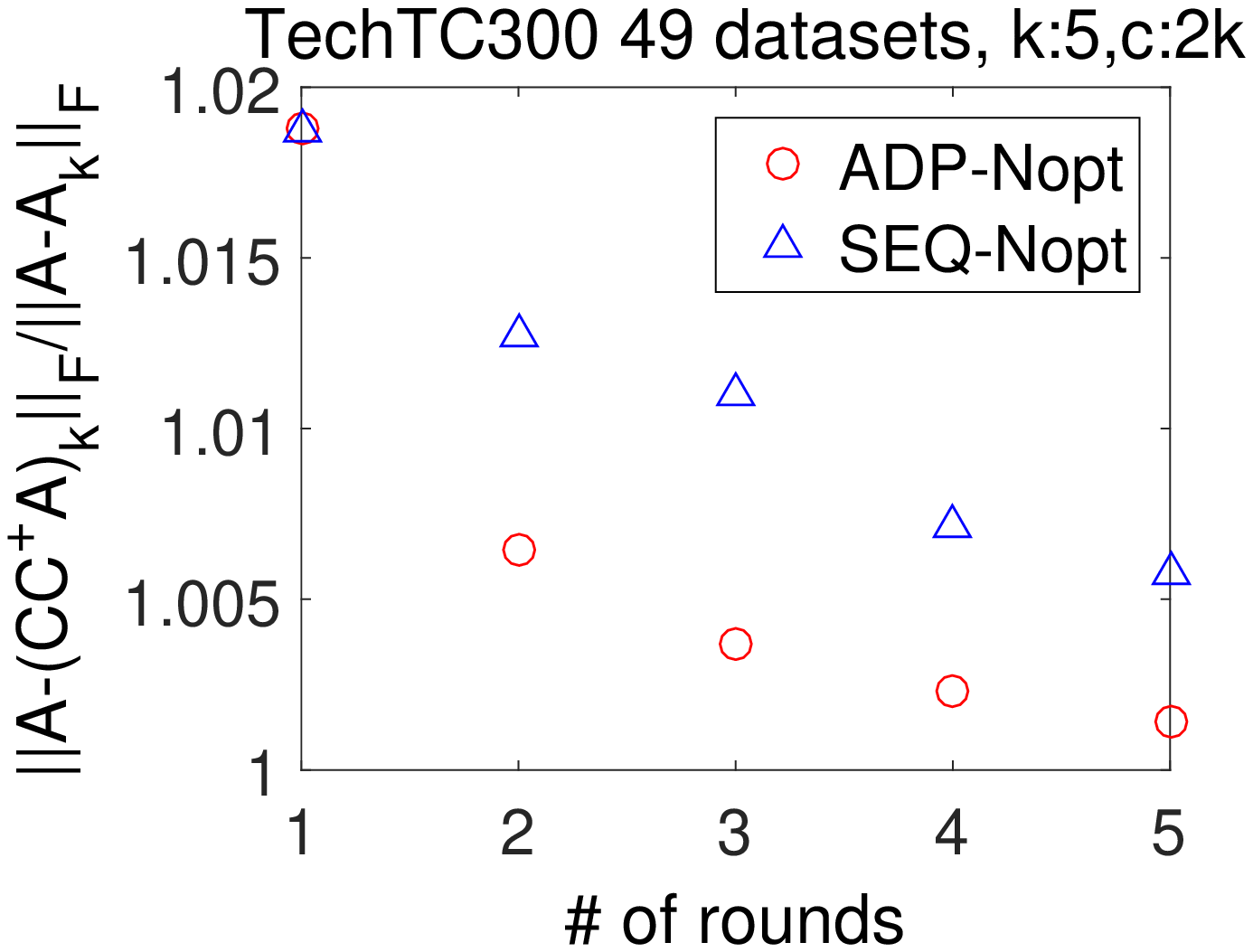}
\end{center}
\caption{The figure on the left shows the error ratio given the same initial selection of columns. The figure on the right shows the error ratio for adaptive sampling vs continued sampling.}
\label{fig:plots3}
\end{figure}
Our next experiment is to test which adaptive strategy works better in 
practice given the same initial selection of columns. That is, in 
Figure~\ref{fig:plots}, {\bf ADP-AE} uses an adaptive sampling based on the
residual \math{\matA-\matC\matC^+\matA} and then adaptively 
samples according to the adaptive strategy in 
\cite{DV06}; the initial columns are chosen with 
the additive error algorithm. Our approach chooses initial columns with 
the relative error CSSP-algorithm and then continues to sample adaptively
based on the relative error CSSP-algorithm and the residual 
\math{\matA-(\matC\matC^+\matA)_{tk}}. We now give
 all the adaptive sampling algorithms the benefit of the 
near-optimal initial columns chosen in the first round by 
the algorithm from~\cite{BMD14}. The result shown in Fig.~\ref{fig:plots3}. 
confirms that {\bf ADP-Nopt} is best even if all adaptive strategies 
\emph{start from the same initial near-optimal columns}.

\noindent \textbf{Adaptive versus Continued Sequential Sampling.}
Our last experiment is to demonstrate that adaptive sampling works
better than continued sequential sampling. We consider 
the relative error CSSP-algorithm in \cite{BMD14} in two modes.
The first is {\bf ADP-Nopt}, which is our adaptive sampling algorithms
which selects \math{tc} columns in \math{t} rounds of \math{c} columns
each. The second is {\bf SEQ-Nopt}, which is just the 
relative error CSSP-algorithm in~\cite{BMD14} sampling \math{tc} columns,
all in one go. The results are shown in Fig.~\ref{fig:plots3}. 
The adaptive boosting of the relative error CSSP-algorithm can gives
up to a 1\% improvement in this data set.

\section{Conclusion}\vskip -0.2cm
We present a new
approach for adaptive sampling algorithms which can boost relative
error CSSP-algorithms, in particular the near optimal CSSP-algorithm 
in~\cite{BMD14}.
We showed theoretical and experimental evidence that our new adaptively boosted
CSSP-algorithm is better than the prior existing adaptive sampling
algorithm which is based on the additive error CSSP-algorithm in~\cite{FKV98}.
We also showed evidence (theoretical and empirical) that our 
adaptive sampling algorithms are better than sequentially sampling
all the columns at once. In particular, our theoretical bounds
give a result which is tighter for matrices whose singular values
decay rapidly.

Several interesting questions remain. We showed that the simplest adaptive
sampling algorithm which samples a constant number of columns in each
round improves upon sequential sampling all at once. 
What is the optimal sampling schedule, and does it depend on the
singular value spectrum of the data matric? In particular,
can improved theoretical bounds or empirical performance 
be obtained by carefully choosing how many columns to select in each round?

It would also be interesting to see the improved adaptive sampling 
boosting of CSSP-algorithms in the actual applications which 
require column selection (such as sparse PCA or unsupervised
feature selection). How do the improved theoretical estimates
we have derived carry over to these problems (theoretically or empirically)?
We leave these directions for future work.

\noindent \textbf{Acknowledgements}\\
Most of the work was done when SP was a graduate student at RPI. PD was supported by IIS-1447283 and IIS-1319280.

\begin{small}
\bibliographystyle{unsrt}
\bibliography{references,mypapers}
\end{small}

\end{document}